\documentclass[twocolumn,journal]{IEEEtran}

\usepackage{graphicx}
\usepackage{bm}
\usepackage{url}
\usepackage{amsfonts,amssymb,amsmath}
\usepackage{mystyle}
\usepackage{nomencl}
\usepackage{cite}
\allowdisplaybreaks                              
\centerfigcaptionstrue
\makeatletter

\newcommand{\Rmnum}[1]{\expandafter\@slowromancap\romannumeral #1@}
\makeatother

\def\Ex{\mathbb E} 
\def\NumFrames{L}
\def\Np {N_p}

\def\Ps {P_s}
\def\Pr {P_r}
\def\tx {\mathbf{s}}
\def\rx {\mathbf{y}}
\def\TR{\mathbf{D}}
\def\rrx {\mathbf{r}_2}
\def\auto {\mathbf{W}}

\def\Hsr {\mathbf{H}}
\def\hsr {\mathbf{h}}
\def\Hrd {\mathbf{F}}

\def\Hrp {\mathbf{G}}
\def\Hrr {\mathbf{H}_r}
\def\Hsp {\mathbf{h}_{p,\ell}}
\def\Heq{\tilde{\mathbf{H}}}
\def\weq{\tilde{\mathbf{w}}}

\begin{document}

\title{Spectrum Sharing with Distributed Relay Selection and
  Clustering} \date{\today} \author{ Yang Li, {\em Student Member,
    IEEE}, and Aria Nosratinia, {\em Fellow, IEEE}\thanks{Manuscript
    approved by E. Serpedin, Editor for Transmission Systems of the
    IEEE Communications Society. Manuscript Received January 21, 2012; Revised June 8, 2012.}
\thanks{The authors are
    with the University of Texas at Dallas, Richardson, TX 75080, USA
    Email: \{yang, aria\}@utdallas.edu.}\thanks{The work
    in this paper was presented in part at the IEEE International
    Conference on Communications (ICC), June 2012. }  }

\maketitle
%%%%%%%%%%%%%%%%%%%%%%%%%%%%%%%%%%%%%%%%%%%%%%%%%%%%%%%%%%%%%%%%%%%
\begin{abstract}
We consider a spectrum-sharing network where $n$ secondary relays are
used to increase secondary rate and also mitigate interference on the
primary by reducing the required overall secondary emitted power. We
propose a distributed relay selection and clustering framework, obtain
closed-form expressions for the secondary rate, and show that
secondary rate increases proportionally to $\log n$. Remarkably, this
is on the same order as the growth rate obtained in the {\em absence}
of a primary system and its imposed constraints.  Our results show
that to maximize the rate, the secondary relays must transmit with
power proportional to $n^{-1}$ (thus the sum of relay powers is
bounded) and also that the secondary source may not operate at its
maximum allowable power.  The tradeoff between the secondary rate and
the interference on the primary is also characterized, showing that
the primary interference can be reduced asymptotically to zero as $n$
increases, while still maintaining a secondary rate that grows
proportionally to $\log n$.  Finally, to address the rate loss due to
half-duplex relaying in the secondary, we propose an alternating relay
protocol and investigate its performance.
\end{abstract}


%%%%%%%%%%%%%%%%%%%%%%%%%%%%%%%%%%%%%%%%%%%%%%%%%%%%%%%%%%%%%%%%%%%
%%%%%%%%%%%%%%%%%%%%%%%%%%%%%%%%%%%%%%%%%%%%%%%%%%%%%%%%%%%%%%%%%
\section{Introduction}
\label{sec:introduction}

Spectrum-sharing~\cite{Ghasemi2007,Gastpar2007} allows unlicensed
(secondary) users to share the spectrum of licensed (primary) users as
long as the interference caused on the primary is tolerable. 
This problem is often formulated as maximizing the secondary rate
subject to interference constraints on the primary, or as the dual
problem of minimizing the interference on the primary subject to a
fixed rate for the secondary. Thus, reducing the interference
footprint of the secondary is of paramount interest in spectrum
sharing.
Multihop relaying and cooperative communication is known to
significantly mitigate interference and increase the sum-throughput in
many multi-user scenarios~\cite{Ozgur2007}, among others in broadcast
channels~\cite{Liang2007}, multiple access channels~\cite{Chen2008}
and interference channels~\cite{Sahin2011}. This has motivated the use
of relays in spectrum sharing
networks~\cite{Zhao2009,Mietzner2009,Asghari2010,Li2011,Naeem2010,Zou2010,Lee2011}.

This paper studies a spectrum sharing network consisting of multiple
primary nodes and a secondary system with $M$-antenna source and
destination, and $n$ half-duplex relays. Unlike conventional relay
networks~\cite{Bolcskei2006,Scaglione2003}, the secondary relays must
not only maximize the secondary rate but also control the interference
on the primary, thus new cooperative algorithms are called for. To
achieve this goal we propose and investigate an approach involving
amplify-and-forward (AF) relaying as well as relay selection. Under
the proposed framework a closed-form expression is derived for the
secondary rate, showing that it increases as $(M\log n)/2$.
Furthermore, we propose an augmented scheduling algorithm that
recovers the half-duplex loss and improves the constant factor in the
throughput growth rate. Finally, we characterize the trade-off between
the secondary rate and the primary interference, showing that the
interference on the primary can be reduced asymptotically to zero
while the secondary rate still grows logarithmically with $n$. Our
results suggest that to maximize the secondary rate subject to primary
interference constraints, one must activate a subset of relays that
are chosen based on their interference profile on the primary, each of
the relays transmit with power inversely proportional to $n$, and
the secondary source must operate at a power level potentially below
its maximum available power. These outcomes are unique to the
cognitive relay networks and are distinct from the conventional relay
networks, e.g.,~\cite{Bolcskei2006}.

Some of the related work is as follows. Zhang~et~al.~\cite{Zhang2009}
studied the secondary power allocation under various power and
interference constraints. The throughput limits of spectrum-sharing
broadcast and MAC were analyzed by Li and
Nosratinia~\cite{Yang2010a}. Recently, relaying in spectrum sharing
networks has attracted attention. For secondary outage probability
Zou~et~al.~\cite{Zou2010} and Lee~et~al.~\cite{Lee2011} proved that
the relay selection in spectrum-sharing achieved the same diversity as
conventional relay networks. For decode-and-forward (DF) relaying,
Mietzner~et~al.~\cite{Mietzner2009} studied power allocation subject
to a desired secondary rate, and Asghari and Aissa~\cite{Asghari2010}
analyzed symbol error rate with relay selection. For AF-relaying,
Li~et~al.~\cite{Li2011} selected a single relay to maximize the
secondary rate, and Naeem~et~al.~\cite{Naeem2010} numerically analyzed
a greedy relay selection algorithm. 
%% None of the past works involving relays
%% in spectrum sharing address the scaling of secondary capacity with the
%% number of relays.

%%%%%%%%%%%%%%%%%%%%%%%%%%%%%%%%%%%%%%%%%%%%%%%%%%%%%%%%%%%%%%%%%%%
%%%%%%%%%%%%%%%%%%%%%%%%%%%%%%%%%%%%%%%%%%%%%%%%%%%%%%%%%%%%%%%%%%%
\section{System Model}
\label{sec:Model}

We consider a spectrum sharing network consists of $\Np$ primary nodes
and a secondary system with an $M$-antenna source, an $M$-antenna
destination and $n$ single-antenna half-duplex relays, as shown in
Figure~\ref{fig:Model}. The average interference power caused by the
secondary on each of the primary nodes must be less than
$\gamma$~\cite{Zhang2009a}. Let $\Hsr\in\mathcal{C}^{M\times n}$ be
the channel coefficient matrix from the source to the relays, and
$\Hrd\in\mathcal{C}^{n\times M}$ and $\Hrp\in\mathcal{C}^{n\times
  \Np}$ be the channel coefficient matrices from the relays to the
destination and the primary nodes, respectively. Denote
$\Hsp\in\mathcal{C}^{M\times 1}$ as the channel vector from the source
to the primary node $\ell$, $1\le \ell \le \Np$.  The source has no
direct link to the destination, a widely used
model~\cite{Naeem2010,Asghari2010,Jing2006,Bolcskei2006} appropriate
for geometries where the relays are roughly located in the middle of
the source and destination.  A block-fading model is considered where
all entries of $\Hsr$, $\Hrd$, $\Hrp$ and $\Hsp$ are zero-mean
i.i.d. circular symmetric complex Gaussian ($\mathcal{CN}$) with
variance $\sigma_s^2$, $\sigma_d^2$, $\sigma_p^2$ and $\sigma_{sp}^2$,
respectively.

\begin{figure}
\centering \includegraphics[width=3.2in]{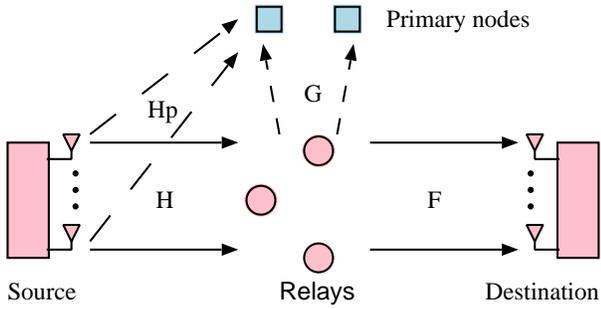}
\caption{System model.}
\label{fig:Model}
\end{figure}

The source communicates with the destination via two hops, which in
general lowers the required transmit power and thus reduces the
interference on the primary. In the first hop, the source sends $M$
independent data streams across $M$ antennas with equal power. The
relay $i$ receives
\begin{equation}
r_i = \sqrt{\frac{\Ps}{M}} \hsr_i^t\,\tx + n_i,
\end{equation}
where $\Ps$ is the source transmit power, which must be less than a
power constraint $\bar{\Ps}$, $\tx\in \mathcal{C}^{M\times 1}$ is
i.i.d. Gaussian signals, $\hsr_i^t\in \mathcal{C}^{1\times M}$ is the
row $i$ of $\Hsr$, namely the channel vector between the relay $i$ and
the source, and $n_i$ is additive noise with distribution
$\mathcal{CN}(0,1)$.

In the second hop, a subset of the relays is selected to transmit to
the destination. We define a random variable $T_i$ to indicate whether
the relay $i$ is selected (eligible):
\begin{equation}
T_i = \begin{cases} 1,  &\mbox{the relay $i$ is eligible}\\ 0, 
\quad &\mbox{otherwise}
\end{cases}.\label{eq:T_i_0}
\end{equation}
No cooperation among the relays is allowed due to their distributed
nature. Each relay rotates and scales $r_i$ by
\begin{equation}
c_i =
e^{j\theta_i}\sqrt{\frac{\Pr}{\Ex[T_i](\Ps\sigma_s^2+1)}}. \label{eq:c_i}
\end{equation}
where $\Pr$ is the average relay power and $\theta_i$ is the rotation
angle, which are designed in the sequel. Therefore, the signal transmitted by
the relay $i$ is
$T_ic_i\,r_i.$
%Note that $t_i$ is non-zero only if the relay $i$ is eligible.

After the relay forwarding, the received signal vector at the
destination is
\begin{align}
\rx & = \sqrt{\frac{\Ps}{M}} \underset{\Heq}{\underbrace{\Hrd \TR
    \Hsr}} \tx + \underset{\weq}{\underbrace{\Hrd \TR \mathbf{n} +
    \mathbf{w}}}, \label{eq:rxd}
\end{align}
where $\TR=\diag(T_1c_1,\cdots,T_nc_n)$ is the relay processing matrix
and $\weq$ is the equivalent additive noise. The equivalent channel
matrix $\Heq$ has entries
\begin{equation}
\big[\Heq]_{mq}  =   \sum_{i=1}^n
T_i\,c_i\,f_{mi}\,h_{iq}, \label{eq:Heq}
\end{equation}
where $f_{mi}$ and $h_{iq}$ are $[\Hrd]_{mi}$ and $[\Hsr]_{iq}$,
respectively.

In this paper, we focus on the effect of the number of relays on the
secondary rate, i.e., the so-called ``scaling laws'' for the relays in
a spectrum-sharing system. Thus, we allow $n$ to increase while $N_p$
remains bounded. Analysis of scaling laws has a long and established
history in wireless communications. Among the many examples we mention
a few, e.g.,~\cite{Gupta2000,Sharif2005,Bolcskei2006}.

We refer to cross channels between secondary transmitters and primary
receivers as {\em interference links}. We assume the destination knows
$\Hrd$, $\TR$ and $\Hsr$, and the relays only know the instantaneous
channel gains to which they directly connect, i.e., $\hsr_i$ and the
column $i$ of $\Hrd$. The interference (thus the channels) from the
primary to the secondary is not explicitly modeled for brevity,
because its impact can be absorbed into the noise term $\tilde{\bw}$.

The cross-channel CSI requirements in a TDD system can be met by the
secondary nodes detecting packets emitted from the primary
nodes. Otherwise, under the spectrum leasing
model~\cite{Jayaweera2009}, the primary nodes can be expected to
actively promote spectrum reuse by transmitting pilots that can be
used for cross-channel gain estimation. The latter model applies to
both TDD and FDD. Regarding the precision of cross-channel CSI, only
the magnitude of the channel gains are needed, and the system can be
made robust to imperfections in the cross-channel CSI to the relays,
as shown in subsequent discussions (see Remark~\ref{remark:CSI}).

%%%%%%%%%%%%%%%%%%%%%%%%%%%%%%%%%%%%%%%%%%%%%%%%%%%%%%%%%%%%%%%%%%%
%%%%%%%%%%%%%%%%%%%%%%%%%%%%%%%%%%%%%%%%%%%%%%%%%%%%%%%%%%%%%%%%%%%
\section{Spectrum-Sharing with relay selection and clustering}
\label{sec:relay}

%% Spectrum sharing has two goals: maximizing the secondary rate and
%% minimizing interference on the primary. 
%% %
%% Relaying improves the performance but also complicates the tradeoff
%% between these two goals. Resolving the intricate tradeoff between the
%% rate versus interference is the key to efficiently using relays in
%% spectrum sharing systems.
%% Therefore, the design of relay schemes is more challenging than
%% that a conventional network. For example, a relay that
%% maximizes the secondary rate may also cause severe interference on the
%% primary.
%
Relays that have weak interference links but strong secondary links
are useful for spectrum sharing, while relays that produce a strong
interference on the primary may do more harm than good. Therefore we
use relay selection. In spectrum sharing, relay selection and
allocation of transmit powers are coupled through the interference
constraint, an issue that is not encountered in conventional
(non-spectrum sharing) relaying.  To make the problem tractable, we
propose a two-step approach: first the allowable interference per
relay is bounded, leading to the creation of an eligible relay
set. Then the secondary rate is maximized by selecting appropriate
relays from among the eligible set and coordinating their
transmissions in a manner shown in the sequel.

%% This decoupling approach will be shown to achieve the same scaling of
%% rate as non-spectrum-sharing networks~\cite{Bolcskei2006}.

%%%%%%%%%%%%%%%%%%%%%%%%%%%%%%%%%%%%%%%%%%%%%%%%%%%%%%%%%%%%%%%%%%%
\subsection{Eligible Relay Selection}
\label{subsec:eligibleRelay}
The interference on the primary nodes is controlled by activating only
the relays with weak interference links. 
We design the relay selection in a distributed manner that does not
require CSI exchange among the relays. A relay is eligible if and only
if all of its own interference link gains are less than a pre-designed
threshold $\alpha$. So from~\eqref{eq:T_i_0}
\begin{equation}
T_i = \begin{cases}
1, \quad  |g_{\ell i}|^2\le \alpha\ \text{for}\ \ell =1,\cdots, \Np\\ 
0, \quad \text{otherwise}
\end{cases}, \label{eq:eligible}
\end{equation}
where $|g_{\ell i}|^2$ is the channel gain between the relay $i$ and
the primary node $\ell$. Note that $\{|g_{\ell i}|^2\}_{\ell,i}$ are
i.i.d. exponentials with mean $\sigma_p^2$, so $\{T_i\}_i$ are
i.i.d. Bernoulli random variables with success probability
\begin{equation}
p = (1-e^{-\alpha/\sigma_p^2})^{\Np}.  \label{eq:eligibleProb}
\end{equation}
%Thus, we have $\Ex[T_i]=p$, $\forall i$.
%\begin{remark}
Since each relay determines eligibility based on its own interference
links, the eligible relay selection is independent across the relays.
%\end{remark}
%
The average interference from the secondary system to the primary node
$\ell$ is
\begin{align}
\gamma_{\ell} & = \frac{1}{2}\Ex\big[(\sum_{i=1}^ng_{\ell i}t_i
  )(\sum_{i=1}^ng_{\ell i}^*t_i^*)\big] +
\frac{\Ps}{2M}\Ex\big[|\Hsp|^2\big] 
%\\ & = \frac{\Pr}{2p}
%\,\sum_{i=1}^n \Ex\big[ |g_{\ell i}|^2 \big | T_i = 1\big]
%\Prob(T_i=1) + \frac{\sigma_{sp}^2\Ps}{2} 
\\ & \ = \frac{\Pr}{2}
  \,\sum_{i=1}^n \Ex\big[ |g_{\ell i}|^2\big | T_i = 1\big]
 +  \frac{\sigma_{sp}^2\Ps}{2} , \label{eq:IntConstraint1}
\end{align}
where the factor $\frac{1}{2}$ is due to the fact that the relays and
the source only transmit during half of the time. The second equality
holds since the design of $\theta_i$ is independent of interference
links, as shown soon. Since $T_i=1$ implies $|g_{\ell i}|^2\le \alpha$
$\forall\ell$, we have
\begin{align}
\Ex\big[ |g_{\ell i}|^2\big | T_i  = 1\big]
& <\int_0^{\alpha}\frac{xe^{-x/\sigma_p^2}}{\sigma_p^2}dx 
 = \sigma_p^2 - e^{-\alpha/\sigma_p^2}(\alpha +
\sigma_p^2 )
\\&\stackrel{\Delta}{=} f(\alpha). \label{eq:ExG}
\end{align}
Combining~\eqref{eq:IntConstraint1} and~\eqref{eq:ExG}, we have
$\forall \ell$, $\gamma_{\ell}\le \gamma$ if $\alpha$ and $\Pr$
satisfy
\begin{equation}
n \Pr f(\alpha) \le \max(\gamma_r,0), \label{eq:IntConstraint}
\end{equation}
where $\gamma_r=2\gamma - \sigma_{sp}^2\Ps$. As long
as~\eqref{eq:IntConstraint} holds, the interference on all the primary
nodes is ensured to be less than $\gamma$, although the relays are
selected distributedly. In our two-hop communication the source power
$\Ps$ is chosen so that $\gamma_r>0$, and otherwise the secondary rate
is zero.

\begin{remark}
\label{remark:CSI}
We briefly discuss CSI uncertainty in the CSI of relay cross-channel
gains. Denote the (relay) estimated cross channel gain as
$|\hat{g}_{\ell i}|^2$. For simplicity, consider $|\hat{g}_{\ell
  i}|^2$ has the same exponential distribution as the true channel
gain $|g_{\ell i}|^2$. Assume uncertainty can be modeled as an
interval, e.g., that the true cross-channel gain is in the interval
$[0, (1+\epsilon)|\hat{g}_{\ell i}|^2]$ for some known and fixed
$\epsilon$. In this case, if $\alpha$ and $\Pr$ satisfy
\begin{equation}
n\Pr f(\alpha +\epsilon) \le \max(\gamma_r,0), \nonumber
\end{equation}
the interference constraints on the primary will still be
ensured. Since $f(\cdot)$ is an increasing and bounded function, the
impact of uncertainty $\epsilon$ is to reduce the transmit power at
the relays.
\end{remark}

%%%%%%%%%%%%%%%%%%%%%%%%%%%%%%%%%%%%%%%%%%%%%%%%%%%%%%%%%%%%%%%
\subsection{Distributed Relay Clustering}
\label{subsec:relayClustering}
The second part of the proposed method aims to maximize the secondary
rate. Recall that the source and destination have $M$ antennas each;
the relays are divided correspondingly into $M$ groups
$\{\mathcal{G}_m \; , \; 1\le m \le M\}$, where each group of relays
aims to provide a virtual pipe between one of the source antennas and
the corresponding destination antenna. This channel-diagonalization
approach is reminiscent of~\cite{Bolcskei2006} but requires more
sophisticated analysis because the (eligible) relay set is random, as
shown in the sequel.

The relay $i\in \mathcal{G}_m$ rotates the received
signal by $\theta_i$ such that
\begin{equation}
 e^{j\theta_i} f_{mi}h_{im}= |f_{mi}||h_{im}|. \label{eq:theta}
\end{equation}
In this case, all the relays in $\mathcal{G}_m$ forward the signal
sent by the source-antenna $m$ coherently to the destination-antenna
$m$.

Now, the challenge is to decide the assignment of relays to the group
$\mathcal{G}_m$, for $1\le m\le M$. We focus on distributed methods so
that the coordination among relays is reduced. %%  Each relay selects a
%% group to join based only on its own channels.
In addition, we
decouple the relay clustering from the relay selection: the relays
decide their groups according to their source-relay and
relay-destination channels but independent of the interference
links. Therefore, under this framework, $\{\theta_i\}_{i=1}^n$ and
$\{T_i\}_{i=1}^n$ are mutually independent. This decoupling allows us
to leverage existing relaying methods to enhance the secondary rate
while bounding the primary interference. It also greatly simplifies
the analysis.

We shall consider two clustering schemes:
\subsubsection{Fixed Clustering}
Here, each of the groups has $n/M$ relays.\footnote{We assume the
  number of relays $n$ is so that $n/M$ is an integer, however, this
  restriction is not essential and can be
  relaxed~\cite{Bolcskei2006}.} Subject to this condition, the relays
are assigned to the groups in a pre-defined manner. Without
loss of generality, we assume: $$ \mathcal{G}_m =
\bigg\{i:\frac{(m-1)n}{M} +1 \le i \le \frac{mn}{M}, \;1\le m\le
M. \bigg\}. $$ 
%%  With this clustering $\{T_i\}_i$ and $\{\theta_i\}_i$
%% are mutually independent.

\subsubsection{Gain Clustering}
In this clustering we have $$ \mathcal{G}_m = \bigg\{i: |h_{im}|>
|h_{iq}|,\;\ q\neq m,\,1\le q\le M\bigg\}.$$ In other words, the groups
are assigned based on the relays' channel gain to source antennas. A
relay (distributedly) decides to join in the group $m$ if its gain to
the $m$-th source antenna is the stronger than any other channel
gains. The group assignment of relays is independent from each other
and is also independent of relay eligibility.
Note that $\mathcal{G}_m$ is no longer fixed but depends on the
source-relay channels. Because all channels are i.i.d., a relay has
equal probability of choosing any of the groups. Therefore
$|\mathcal{G}_m|$ (the cardinality of $\mathcal{G}_m$) is
binomially distributed with parameters $(n,\frac{1}{M})$.

%%%%%%%%%%%%%%%%%%%%%%%%%%%%%%%%%%%%%%%%%%%%%%%%%%%%%%%
%%%%%%%%%%%%%%%%%%%%%%%%%%%%%%%%%%%%%%%%%%%%%%%%%%%%%%%
\section{Secondary Rate in Spectrum-sharing with Relays}
\label{sec:rate}

We first derive a general closed-form expression for the secondary
rate under the proposed framework, and then evaluate the achievable
rate for specific methods.

\subsection{Calculation of Secondary Rate}
From~\eqref{eq:rxd}, conditioned on $\Hrd$, $\TR$ and $\Hsr$, $\weq$
is a Gaussian vector with autocorrelation
\begin{equation}
\mathbf{W} = \bI + \Hrd \TR\TR^{\dag} \Hrd^{\dag}. \label{eq:weq}
\end{equation}
The secondary rate in the presence of $n$ relays is denoted with $R_n$
and is given by:
\begin{equation}
R_n = \frac{1}{2}\log \det\bigg( \bI+\frac{\Ps}{M}
\Heq\Heq^{\dag}\mathbf{W}^{-1}\bigg), \label{eq:eqNoiseAutoC}
\end{equation}
where $\frac{1}{2}$ is due to the half-duplex relay constraint.

Now, we find $R_n$ for large $n$. First, from~\eqref{eq:Heq}
and~\eqref{eq:theta}, the entry of $\Heq$ is
\begin{align}
\big[\Heq]_{mq} & = \begin{cases} A_{mm} + B_{mm}, & \mbox{$q=m$} 
\\ C_{mq}, & \mbox{$q\neq m$}
\end{cases},
\end{align}
where
\begin{align}
A_{mm} & = \sqrt{\frac{\Pr}{p(\sigma_s^2\Ps+1)}}\sum_{i\in
    \mathcal{G}_m} 
  T_i\,|f_{mi}|\,|h_{im}|,\nonumber
\\ B_{mm}& = \sqrt{\frac{\Pr}{p(\sigma_s^2\Ps+1)}} \sum_{i\notin \mathcal{G}_m} T_i\,f_{mi}\,h_{im}\,e^{j\theta_i},\nonumber\\
C_{mq} & = \sqrt{\frac{\Pr}{p(\sigma_s^2\Ps+1)}} \sum_{i=1}^n
    T_i\,f_{mi}\,h_{iq}\,e^{j\theta_i}. 
\end{align}
The terms in $A_{mm}$, $B_{mm}$ and $C_{mq}$ are mutually independent,
because $\{T_i\}_{i=1}^n$ and $\{\theta_i\}_{i=1}^n$ are independent
from each other. So we have the following lemma.

\begin{lemma}
\label{lemma:Heq}
If $\min_{1\le m\le M} |\mathcal{G}_m|
\stackrel{w.p.1}{\longrightarrow}\infty$ as $n\rightarrow\infty$, we
have
\begin{align}
\frac{A_{mm}}{n} - \frac{1}{n}\sqrt{\frac{p\Pr}{\sigma_s^2\Ps+1}}\sum_{i\in
  \mathcal{G}_m} \Ex[|f_{mi}||h_{im}|]
&\stackrel{w.p.1}{\longrightarrow} 0,\label{eq:Amm}\\
\frac{B_{mm}}{n} - \frac{1}{n}\sqrt{\frac{p\Pr}{\sigma_s^2\Ps+1}}\sum_{i\notin \mathcal{G}_m}
\Ex[f_{mi}h_{im}e^{-j\theta_i}]
&\stackrel{w.p.1}{\longrightarrow} 0, \label{eq:Bmm}\\
\frac{C_{mq}}{n} -\frac{1}{n}\sqrt{\frac{p\Pr}{\sigma_s^2\Ps+1}}\sum_{i=1}^n
\Ex[f_{mi}h_{iq}e^{-j\theta_i}]
&\stackrel{w.p.1}{\longrightarrow}  0 . \label{eq:Cmq}
\end{align}
\end{lemma}
\begin{proof}
The proof follows from~\cite[Theorem 2.1]{Gutbook} and~\cite[Theorem
  1.8.D]{Serflingbook}, and is omitted here. 
\end{proof}

From Lemma~\ref{lemma:Heq}, given
$|\mathcal{G}_m|\stackrel{w.p.1}{\longrightarrow}\infty$ $\forall m$,
we have:
\begin{equation}
\frac{\Heq}{n}- \diag(a_1,\cdots,a_M)
\stackrel{w.p.1}{\longrightarrow} 0, \label{eq:HeqDiag}
\end{equation}
where
\begin{equation}
a_m = \frac{1}{n}\sqrt{\frac{p\Pr}{(\sigma_s^2\Ps+1)}}\sum_{i\in \mathcal{G}_m}
\Ex[|f_{mi}||h_{im}|] .\label{eq:am}
\end{equation}

The above analysis indicates that $\Heq$ converges to a diagonal
matrix for large $n$ (with probability 1). We now show that $\mathbf{W}$ is
also diagonalized as $n$ increases. From~\eqref{eq:weq}, we have
\begin{equation}
\big[\mathbf{W}\big]_{mq} = \sum_{i=1}^n\frac{T_i\Pr
  f_{mi}f_{iq}^*}{p(\Ps\sigma_s^2+1)} + \delta_{mq},
\end{equation}
where $\delta_{mq}=1$ if $m=q$ and $\delta_{mq}=0$ if $m\neq q$. One
can verify Kolmogorov conditions~\cite[Theorem 1.8.D]{Serflingbook},
and therefore obtain
\begin{align}
\frac{\big[\mathbf{W}\big]_{mq}}{n} -
\frac{1}{n} \bigg(\sum_{i=1}^n\frac{\Pr\Ex
  [f_{mi}f_{iq}^*]}{\Ps\sigma_s^2+1} +\delta_{mq}\bigg) 
& \stackrel{w.p.1}{\longrightarrow} 0,
\end{align}
where
\begin{equation}
\Ex[f_{mi}f_{iq}^*] = \begin{cases}
  \Ex[|f_{mi}|^2], &\mbox{$m=q$ }\\
0, & \mbox{$m\neq q$}
\end{cases}.
\end{equation}
Therefore, we have
\begin{equation}
\frac{\mathbf{W}}{n} - \diag(b_1,\cdots,b_M)
\stackrel{w.p.1}{\longrightarrow} 0, \label{eq:WeqDiag}
\end{equation}
where
\begin{equation}
b_m = \frac{\Pr \sum_{i=1}^n\Ex[|f_{mi}|^2]}{n(\Ps\sigma_s^2+1)} +
\frac{1}{n}. \label{eq:bm}
\end{equation}

%% in the above $\Ex[|f_{mi}|^2]$ is identical for $1\le i\le n$,
%% since the relays have the same selection rule that leads to identical
%% statistics of $|f_{mi}|^2$.

From~\eqref{eq:HeqDiag} and~\eqref{eq:WeqDiag}, for large $n$, the
end-to-end channel between the source and the destination is
approximately decoupled into $M$ parallel channels under the proposed
framework, where the channel coefficient $m$ is $a_m$ and the received
noise has variance $b_m$. The capacity of this parallel channel is
\begin{equation}
\overline{R} = \frac{1}{2}\sum_{m=1}^M\log\big(1+\frac{n\Ps a_m^2}{Mb_m}\big),\label{eq:R_n}
\end{equation}
Therefore, it is reasonable to expect that $R_n\approx \overline{R}$
for large $n$. After some calculation (omitted for brevity),
we obtain the following result.
\begin{theorem}
\label{theorem:R_n}
Consider a secondary system with an $M$-antenna source, an $M$-antenna
destination, and $n$ single-antenna relays, in the presence of $N$
primary nodes each tolerating interference no more than $\gamma$. The
secondary rate satisfies
\begin{equation}
R_n-\overline{R} \stackrel{w.p.1}{\longrightarrow}  0,\quad
n\rightarrow \infty,\label{eq:asconvergence}
\end{equation}
under the proposed relay selection and clustering framework.
\end{theorem}

\subsection{Achievable Rate under Specific Clustering Schemes} 

We apply Theorem~\ref{theorem:R_n} to fixed clustering and
gain clustering.

\subsubsection{Fixed Clustering}
In this scheme, $|\mathcal{G}_m|=\frac{n}{M}$ (so
Lemma~\ref{lemma:Heq} is applicable), and $|f_{mi}|$ and $|h_{im}|$
are i.i.d. Rayleigh random variables with mean
$\frac{\sigma_d\sqrt{\pi}}{2}$ and $\frac{\sigma_s\sqrt{\pi}}{2}$,
respectively. Therefore, from~\eqref{eq:am}, 
$
a_m = \frac{\pi\sigma_s\sigma_d}{4M}\sqrt{\frac{p\Pr}{\sigma_s^2\Ps+1}},
$
for $1\le m \le M$. Under this clustering, $|f_{mi}|^2$ is
i.i.d. exponential with mean $\sigma_d^2$, and we have
$
b_m = \frac{\sigma_d^2\Pr}{\sigma_s^2\Ps+1}+\frac{1}{n}, 
$
for $1\le m \le M$. Substituting $a_m$ and $b_m$ into~\eqref{eq:R_n},
$\overline{R}$ becomes
\begin{equation}
R^{(f)} = \frac{M}{2}\log\bigg( 1+ \frac{np
  \pi^2\sigma_s^2\sigma_d^2\Pr\Ps}{16M^3(\sigma_d^2\Pr+n^{-1}(\sigma_s^2\Ps+1))}
\bigg). \label{eq:R^{(f)}}
\end{equation}
From Theorem~\ref{theorem:R_n}, under fixed clustering, we have:
$
R_n - R^{(f)} \stackrel{w.p.1}{\longrightarrow} 0.
$

\subsubsection{Gain Clustering}
Since $|\mathcal{G}_m|$ is binomially distributed with parameters
$(n,\frac{1}{M})$, we have
$|\mathcal{G}_m|/n\stackrel{w.p.1}{\longrightarrow} 1/M$, and
Lemma~\ref{lemma:Heq} is again applicable. Due to the independence of
$|f_{mi}|$ and $|h_{im}|$, from~\eqref{eq:am}, we have
\begin{equation}
a_m = \frac{1}{n}\sqrt{\frac{p\Pr}{(\sigma_s^2\Ps+1)}}\sum_{i\in \mathcal{G}_m}
\Ex[|f_{mi}|]\Ex[|h_{im}|] .  \label{eq:a_mGain}
\end{equation}
where $\Ex[|f_{mi}|]=\frac{\sigma_d\sqrt{\pi}}{2}$ (i.i.d. Rayleigh)
and $\Ex[|h_{im}|]=\max_{1\le m \le M}|h_{mi}|$, which is the maximum
of $M$ i.i.d. Rayleigh random variables. We have
\begin{align}
\mu_h& = \Ex\big[\max_{1\le m\le M} |h_{im}|\big] \nonumber\\
&= \int_0^{\infty}
\frac{2Mx^2}{\sigma_s^2}e^{-x^2/\sigma_s^2}\big(1-e^{-x^2/\sigma_s^2}\big)^{M-1}\,dx\\
& = \sum_{m=0}^{M-1}(-1)^{M-m-1}\big(_{\ \ m}^{M-1}\big)
\frac{\sigma_sM\Gamma(\frac{3}{2})}{(M-m)^{3/2}}. \label{eq:mu_h}
\end{align}
Note that $\mu_h=\frac{\sigma_s\sqrt{\pi}}{2}$ for $M=1$ (no selection
is needed), which is identical to the fixed clustering. Based
on~\eqref{eq:a_mGain} and $|\mathcal{G}_m|/n
\stackrel{w.p.1}{\rightarrow} 1/M$, we have
$
a_m- \frac{\sigma_d\mu_h}{2M}\sqrt{\frac{p\pi\Pr}{(\sigma_s^2\Ps+1)}}
\stackrel{w.p.1}{\longrightarrow} 0.
$

Under this clustering, $b_m$ remains the same as the fixed clustering
case, since $|f_{mi}|$ is still i.i.d. Rayleigh for $i\in
\mathcal{G}_m$, $\forall m$. Substituting $a_m$ and $b_m$
into~\eqref{eq:R_n}, we have
\begin{equation}
R^{(g)} =  \frac{M}{2}\log\bigg( 1+ \frac{np
  \pi\mu_h^2\sigma_d^2\Pr\Ps}{4M^3(\sigma_d^2\Pr+n^{-1}(\sigma_s^2\Ps+1))}\bigg),\label{eq:R^{(g)}}
\end{equation}
then:
$
R_n - R^{(g)} \stackrel{w.p.1}{\longrightarrow} 0.
$

%%%%%%%%%%%%%%%%%%%%%%%%%%%%%%%%%%%%%%%%%%%%%%%%%%%%%%%%%%%%%%%%%%
\section{Optimal Power Strategy for Spectrum-sharing with relays}
\label{sec:optimal}

In general, one may envision two competing philosophies for relay
selection: (1) Allow only relays that have extremely weak interference
links to the primary. Only very few relays will qualify but each of
them can transmit at high power.
(2) Allow a large number of relays to be activated. In this case the relay powers
must be lowered because not all interference links are as ``good'' as
the previous case. 

The key question is: which approach is better? Should we use a few select relays with
excellent interference profiles, or more relays operating at lower
power?
In this section, we optimize the threshold $\alpha$, the relay power
$\Pr$ and the source power $\Ps$, while bounding the primary
interference. The results of this section show that in general the
balance tips in favor of having more eligible relays operating at low
power.

\subsection{Optimal Design of $\alpha$ and $\Pr$}

Consider a fixed $P_s$. Since $\alpha$
and $\Pr$ depend on each other via~\eqref{eq:IntConstraint}, given
$\alpha$ the maximum $\Pr$ is
\begin{equation}
\Pr=\frac{\gamma_r}{nf(\alpha)}.   \label{eq:Pr}
\end{equation}
Substituting~\eqref{eq:Pr} and~\eqref{eq:eligibleProb}
into~\eqref{eq:R^{(f)}} and~\eqref{eq:R^{(g)}} shows that $R^{(f)}$
and $R^{(g)}$ attain their maxima (as a function of $\alpha$) at
$\alpha=\alpha_0$ where:
\begin{equation}
\alpha_o = \arg\max_{\alpha}\frac{\gamma_r\Ps(1-e^{-\alpha/\sigma_p^2})^{\Np}}{\gamma_r \sigma_d^2 +
  (\sigma_s^2\Ps+1)f(\alpha)}. \label{eq:optAlpha}
\end{equation}

A closed-form solution for $\alpha_o$ is unavailable but numerical
solution can be easily obtained. Figure~\ref{fig:alpha} shows the
optimal design of $\alpha$ based on~\eqref{eq:optAlpha}. For both fixed clustering and
gain clustering, according to~\eqref{eq:optAlpha}, $\alpha_o=1.7$
maximizes the secondary rate. 

Now, we characterize the asymptotic behavior of $\alpha_o$,
equivalently the optimal $\Pr$. Because~\eqref{eq:optAlpha} is
independent of $n$, the optimal threshold $\alpha$ is not a function
of $n$. So from~\eqref{eq:Pr} the optimal average transmit
power\footnote{It can be shown that Theorem~\ref{theorem:R_n} still
  holds if $\Pr$ scales as $\Theta(n^{-1})$.} is $\Pr=\Theta(n^{-1})$,
i.e., there exist real constants $d_1,d_2>0$ so that $ d_1 n^{-1}\le \Pr
\le d_2 n^{-1}$. This implies that the secondary system should on
average allow many relays to operate at low power. One may intuitively
interpret this result as follows. To comply with the primary
interference constraints, the sum power of relays must be bounded, and
by spreading the total power among more relays better beamforming gain
is achieved via coherent transmission.

\begin{figure}
\centering \includegraphics[width=3.5in]{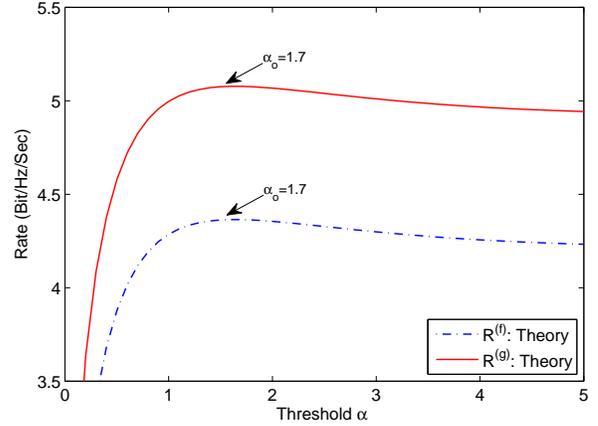}
\caption{Optimal value of selection threshold $\alpha$ under $P_s=5, n=100$}
\label{fig:alpha}
\end{figure}

Now, we study the scaling of the secondary rate. Consider examples
with fixed clustering (Eq.~\eqref{eq:R^{(f)}}) and gain clustering
(Eq.~\eqref{eq:R^{(g)}}). If $\Pr=\Theta(n^{-1})$, for fixed $\alpha$
(not necessarily optimal), we have
\begin{equation}
R^{(f)} = \frac{M}{2}\log n + C_1, \quad R^{(g)} = \frac{M}{2}\log n + C_1 +
C_2, \label{eq:RScale}
\end{equation}
where $$C_1=\frac{M}{2}\log\frac{
  p\pi^2\sigma_s^2\sigma_d^2\gamma_r\Ps} {16M^3 (\sigma_s^2\Ps+1)
  (\sigma_p^2\gamma_r + f(\alpha))}$$ and
$C_2=\log\frac{4\mu_h^2}{\pi\sigma_s^2}$. One can view $C_2$  as
multi-antenna diversity gain by selecting over source-relay
channels. From~\eqref{eq:RScale}, the secondary rate increases as
$(M\log n)/2$, which is summarized in the next theorem.

\begin{theorem}
\label{theorem:RScale}
Consider a secondary system with an $M$-antenna source, an $M$-antenna
destination, and $n$ single-antenna relays, in the presence of $N$
primary nodes each tolerating interference no more than $\gamma$. For
$\Pr=\Theta(n^{-1})$ and fixed $\alpha$, the secondary rate satisfies
\begin{equation}
\frac{R_n}{\frac{M}{2}\log n} \stackrel{w.p.1}{\longrightarrow} 1,
\end{equation}
under the proposed framework with both fixed clustering and  gain
clustering.
\end{theorem}

Theorem~\ref{theorem:RScale} holds for a broad class of clustering
schemes, as long as the corresponding $a_m$ and $b_m$ are bounded but
non-zero, i.e., the secondary end-to-end equivalent channel is diagonalized
with probability 1 as $n$ grows.

\begin{remark}
\label{remark:peak}

It is possible to extend our results to the case of peak interference
constraint $\gamma$. The secondary source will manage its instantaneous
interference to be smaller than $\gamma_s$ on all primary nodes by
adjusting its transmit power according to the largest cross-channel
gain to the primaries. Then, the sum interference from all the relays must be smaller
than $\gamma_r = \gamma - \gamma_s$. Let $\Pr = \xi/n$ where $\xi$ is a
positive constant. The instantaneous interference from all the relays to the
primary node $\ell$ is $\gamma_{\ell} = \xi \sum_{i=1}^nT_i |g_{\ell
  i}|^2 /n$. This implies that
\begin{align}
\gamma_{\ell} - \xi \; \Ex[T_i \;|g_{\ell i}|^2 ]   \stackrel{w.p.1}{\longrightarrow} 1\nonumber
\end{align}
for an arbitrary $ i\in\{1,\ldots,n\}$, where we have used the fact
that $T_i$ and $g_{\ell i}$ for all $i$
have identical distributions. Therefore, $\xi= \gamma_r \big( \Ex[T_i
  \;|g_{\ell i}|^2 ] \big)^{-1}$ ensures the instantaneous interference
on all the primary nodes to be smaller than $\gamma$ with probability
$1$.
\end{remark}

\subsection{Optimal Source Power}

Due to the primary interference constraints, for any chosen $\alpha$
the higher the source power $\Ps$, the lower the relay power $\Pr$,
and vice versa. From~\eqref{eq:R^{(f)}} and~\eqref{eq:R^{(g)}}, the
rate-maximizing $\Ps$ is
\begin{equation}
\Ps^* = \arg\max_{0<\Ps <
\min(\bar{\Ps},2\gamma)}\frac{(2\gamma - \sigma_{sp}^2 \Ps)\Ps}{ (2\gamma - \sigma_{sp}^2\Ps)\sigma_d^2 +
  (\sigma_s^2\Ps+1)f(\alpha)}. 
\end{equation}
The unique solution of the above optimization problem is 
\begin{equation}
\Ps^* = \min(P_o,\bar{\Ps}), \label{eq:optPs}
\end{equation}
where $P_0 = \frac{\gamma}{\sigma_{sp}^2} $ if
$\sigma_{sp}^2\sigma_d^2=\sigma_s^2f(\alpha)$, otherwise:
\begin{align}
P_o= &
\frac{2\gamma\sigma_d^2}{\sigma_{sp}^2\sigma_d^2-\sigma_s^2f(\alpha)}
\nonumber\\& -
  \frac{
    \sqrt{ \big(2\gamma\sigma_d^2f(\alpha)+f^2(\alpha)\big)\big(\sigma_{sp}^2+2\gamma\sigma_s^2\big)}}{\sigma_{sp}\big(\sigma_{sp}^2\sigma_d^2-\sigma_s^2f(\alpha)\big)}
\end{align}

Figure~\ref{fig:optPsTheory} demonstrates the optimal source power as
a function of three channel parameters $\sigma_{sp}^2$, $\sigma_s^2$
and $\sigma_d^2$. Three curves are shown, in each case one parameter
varies while the other two are held constant (at unity). In this
Figure $\bar{\Ps}=10$, $\gamma=5$ and $f(\alpha)=0.8$.  As the
source-primary channels become stronger, the source needs to reduce
power; otherwise, the relay power must decrease to comply with the
primary interference constraints, which curbs the rate achieved by the
second hop. If the source-relay channels become stronger, the
relay-destination links is the bottleneck and the relays need to
transmit at higher power, thus once again the source needs to reduce
power. In contrast, when the relay-destination channels become better,
the source-relay channels are the bottleneck so the source needs to
increase power.

\begin{figure}
\centering
\includegraphics[width=3.5in]{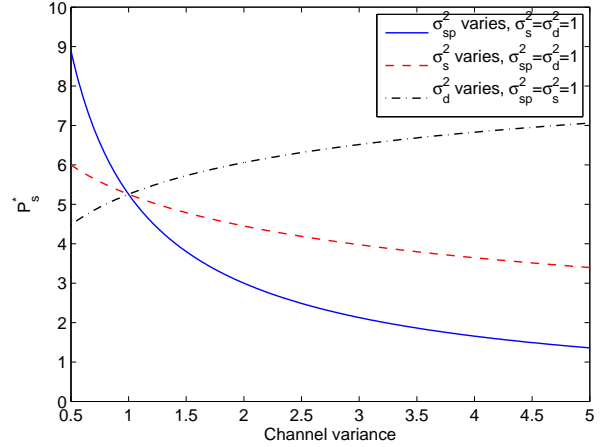}
\caption{Optimal source power with $\gamma=5, f(\alpha)=0.8$}
\label{fig:optPsTheory}
\end{figure}

\subsection{Asymptotic Reduction of Interference on Primary}
Multiple relays produce opportunities not only to enhance the
secondary rate but also to reduce the interference on the
primary. Suppose the interference on the primary nodes to be bounded
as $ \gamma = O\big(n^{-\delta}\big), $ which goes to zero as
$n\rightarrow\infty$. From~\eqref{eq:IntConstraint}, it is sufficient
to comply with this constraint if $\Pr$ decreases as
$\Theta(n^{-(1+\delta)})$ and $\Ps$ decreases as
$\Theta(n^{-\delta})$. Substituting $\Pr$ and $\Ps$ into the
expression of $R^{(f)}$ given by~\eqref{eq:R^{(f)}} and following some
order calculation (the analysis of $R^{(g)}$ is the same thus
omitted), we have
\begin{equation}
R^{(f)} = \begin{cases} 
\frac{M(1-2\delta)}{2}\log n + O(1) &\mbox{$\delta <\frac{1}{2}$}
\\ O(1) &\mbox{$\delta \ge \frac{1}{2}$}
\end{cases}
\end{equation}

The above equation characterizes the trade-off between the secondary
rate and the interference on the primary: the faster of the
interference reduction, the slower of the rate growth. It also shows
that the interference on the primary nodes may be mitigated (to zero
asymptotically), while the secondary rate maintains to increase as
$\Theta(\log n)$.
\begin{remark}
In the above, the allowable interference $\gamma$ is made to decline
as $\Theta(n^{-\delta})$, which leads the growth rate to decrease
linearly in $\delta$. If $\gamma$ is reduced more slowly, e.g.,
decreasing as $\Theta(\frac{1}{\log n})$, the secondary rate can
increase at a rate of $\frac{M}{2}\log n$. If we try to mitigate the
primary interference faster than $\Theta(1/\sqrt{n})$, the secondary
rate will not increase logarithmically with $n$.
\end{remark}

%%%%%%%%%%%%%%%%%%%%%%%%%%%%%%%%%%%%%%%%%%%%%%%%%%%%%%%%%%%%%%%%%%%
%%%%%%%%%%%%%%%%%%%%%%%%%%%%%%%%%%%%%%%%%%%%%%%%%%%%%%%%%%%%%%%%%%%
\section{Spectrum-sharing with Alternating Relay Protocol}

In this section we consider issues raised by the relay half-duplex
constraint, i.e., limitations that arise because relays cannot listen
to the source at the same time as they are transmitting. When a subset
of relays are activated for relaying the previously received
information, the inactive relays are able to listen and receive
information from the source, thus in principle the source can transmit
continually and the half-duplex loss can be mitigated. This is the
basic idea of {\em spectrum sharing with Alternating Relay Protocol},
which is the subject of this section.

The protocol consists of $\NumFrames$ transmission frames, as shown in
Figure~\ref{fig:Protocol}. It is assumed the channel coefficient
remains constant during each frame, but varies independently from
frame to frame. The source transmits during frames $1$ through
$\NumFrames-1$, and remains silent during frame $\NumFrames$. Since
the source transmits $\NumFrames-1$ data segments during $\NumFrames$
time intervals, the rate loss induced by the half-duplex relaying is a
factor of $\frac{\NumFrames-1}{\NumFrames}$. The relays are
partitioned into two groups $\mathcal{G}_1=\{1\le i\le \frac{n}{2}\}$
and $\mathcal{G}_2=\{\frac{n}{2}+1\le i\le n\}$. During even-numbered
transmission frames a subset of the relays in $\mathcal{G}_1$ transmit
to the destination, while the relays in $\mathcal{G}_2$ listen to the
source. During odd-numbered transmission frames, a subset of the
relays in $\mathcal{G}_2$ transmit, while the relays in
$\mathcal{G}_1$ listen.  As shown later, each of the two relay groups
asymptotically achieves a rate that grows as
$\frac{M}{2}\frac{\NumFrames-1}{\NumFrames}\log n$, thus the overall
system has a rate that grows proportionally to
$M\frac{\NumFrames-1}{\NumFrames}\log n$. Therefore a good part of the
half-duplex rate loss can be recovered.

\begin{figure} 
\centering 
\includegraphics[width=3.5in]{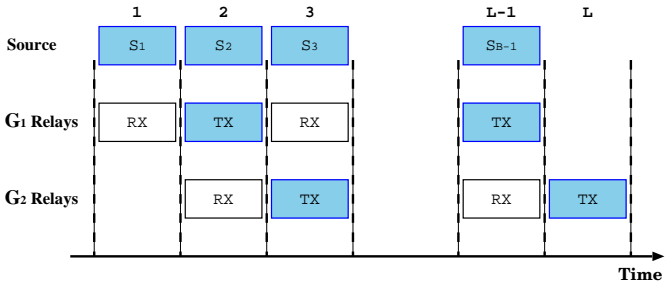}
\caption{Transmission schedule in the alternating relay protocol (ARP) }
\label{fig:Protocol}
\end{figure}

When either group
$\mathcal{G}_1$ or group $\mathcal{G}_2$ is in the transmit mode, a
subset of relays in the corresponding group is selected to 
transmit. A relay is selected (eligible) if its interference links
satisfy~\eqref{eq:eligible}, similar to
Section~\ref{subsec:eligibleRelay}. The average interference power on
the primary node $\ell$ takes slightly different forms depending on
whether $\NumFrames$ is even or odd. When $\NumFrames$ is even:
\begin{align}
\gamma_{\ell} = &\frac{(\NumFrames-1)\sigma_{sp}^2\Ps}{\NumFrames}+ \frac{\Pr}{2}\sum_{i\in\mathcal{G}_1} \Ex\big[ |g_{\ell i}|^2\big
    | T_i = 1\big]  \nonumber\\
&\qquad +\frac{(\NumFrames-2)\Pr}{2\NumFrames}\sum_{i\in\mathcal{G}_2} \Ex\big[ |g_{\ell i}|^2\big
    | T_i = 1\big],
\end{align}
and when $\NumFrames$ is odd:
\begin{align}
\gamma_{\ell} =&\frac{(\NumFrames-1)\sigma_{sp}^2\Ps}{\NumFrames}+\frac{(\NumFrames-1)\Pr}{2\NumFrames}
  \bigg(\sum_{i\in\mathcal{G}_1} \Ex\big[ |g_{\ell i}|^2\big
    | T_i = 1\big] \nonumber\\
& \qquad +\sum_{i\in\mathcal{G}_2} \Ex\big[ |g_{\ell i}|^2\big
    | T_i = 1\big] \bigg) \; .
\end{align}
To comply with the interference constraints on the primary nodes, the
threshold $\alpha$ and the relay power $\Pr$ shall satisfy
\begin{equation}
n \Pr f(\alpha)
\le  \max(\gamma_L,0), \label{eq:IntConstraint3}
\end{equation}
where $\gamma_\NumFrames= \frac{2\NumFrames}{\NumFrames-1}\gamma -
2\sigma_{sp}^2\Ps $ with $\Ps$ so that $\gamma_\NumFrames>0$, and we
use the fact that $\Ex\big[ |g_{\ell i}|^2\big | T_i =
  1\big]=f(\alpha)$. Since from Section~\ref{sec:optimal} the optimal
$\Pr$ is proportional to $n^{-1}$, we let $\Pr=\eta/n$, and
re-write~\eqref{eq:IntConstraint3} as
\begin{equation}
\eta f(\alpha)
\le  \max(\gamma_L,0). \label{eq:IntConstraint4}
\end{equation}

For the Alternating Relay Protocol, relay clustering is accomplished
in a manner similar to Section~\ref{subsec:relayClustering}, therefore
the details are omitted. During frame $2k$ (or $2k+1$), let
$\mathcal{G}_{m,k}^{(1)}\subset \mathcal{G}_1$ (or
$\mathcal{G}_{m,k}^{(2)}\subset \mathcal{G}_2$) be the set of relays
that assists the antenna pair $m$. As long as $\min_{m,k,d}
|\mathcal{G}_{m,k}^{(d)}|\rightarrow \infty$, the secondary rate will
be obtained following the analysis similar to Section~\ref{sec:rate}.

\begin{remark}
\label{remark:general}
At any point in time, it is possible to allow {\em all}
non-transmitting relays to listen to the source, and be eligible to
transmit in the next frame. This may give some gains, however, it also
complicates the relay selection by introducing dependence between not
only interference links but also other links such as source-relay and
relay-relay links. It may be better for a relay even with a small
interference on primary to remain inactive if it has also a weak
channel to destination (therefore it cannot help much) but has a
strong channel to the source (therefore it can listen well for the
next round). Thus, any gains will come with a loss of elegance and
tractability, and therefore this approach is not considered in this
paper.
\end{remark}

\subsection{A Simple Example: $\NumFrames=3$}

For illustration purposes, we consider $\NumFrames=3$, where $\mathcal{G}_1$
($\mathcal{G}_2$) listen to the source during frame $1$, and then
transmit to the destination during frame $2$ (frame $3$). We assume
fixed clustering is used with $|\mathcal{G}_{m,1}^{(d)}|=n/(2M)$, for
$1\le m\le M$ and $1\le d\le 2$. Let $\Hsr_d$ ($\Hrd_d$) be the
channel coefficient matrix between the relays in $\mathcal{G}_d$ and
the source (the destination), and $\Hrr$ be the channel coefficient
matrix between $\mathcal{G}_1$ and $\mathcal{G}_2$ with
i.i.d. $\mathcal{CN}(0,\sigma_r^2)$ entries.

%% All channels are assumed known by the destination.\footnote{Comments
%%   on this.}

We now analyze the rate achieved under Alternating Relay Protocol. The
optimization of the threshold and the source power follows in a manner
similar to Section~\ref{sec:optimal} and thus is omitted here.

\subsubsection{Rate Achieved by $\mathcal{G}_1$}
After listening to the source at frame $1$, $\mathcal{G}_1$ relays to
the destination at frame $2$. At the end of frame $2$, similar
to~\eqref{eq:rxd} the received signal at the destination is
\begin{align}
\rx_1 & = \sqrt{\frac{\Ps}{M}} \underset{\Heq_1}{\underbrace{\Hrd_1
    \TR_1 \Hsr_1}} \tx_1 + \underset{\weq_1}{\underbrace{\Hrd_1 \TR_1
    \mathbf{n}_1 + \mathbf{w}_1}}, \label{eq:rxd1}
\end{align}
where $\tx_1$ is the signal sent by the source during frame $1$,
$\mathbf{n}_1$ is the noise forwarded by the group $\mathcal{G}_1$ of
relays, $\mathbf{w}_1$ is the destination noise. For the group
$\mathcal{G}_1$ the relay gains are collected into the relay
processing matrix
\begin{equation}
\TR_1 =\diag(T_1c_1,\cdots,T_{\frac{n}{2}}c_{\frac{n}{2}}),
\end{equation}
where $c_i$ is given by~\eqref{eq:c_i} so
that the average relay power constraints are satisfied. One can verify
that the equivalent channel $\Heq_1$
\begin{equation}
\frac{\Heq_1}{\sqrt{n}}  \stackrel{w.p.1}{\longrightarrow} \rho_1 \bI,
\end{equation}
where $\rho_1=
\frac{\pi\sigma_s\sigma_d}{8M}\sqrt{\frac{p\eta}{\sigma_s^2\Ps+1}}$. The
auto-covariance of equivalent noise $\weq_1$ is
\begin{equation}
\frac{1}{n}\auto_1  \stackrel{w.p.1}{\longrightarrow} \lambda_1 \bI,
\end{equation}
where $\lambda_1=
\frac{\eta\sigma_d^2}{2(\sigma_s^2\Ps+1)}+1$. Therefore, the
end-to-end channel is diagonalized for large $n$, and similar to the
results in Theorem~\ref{theorem:R_n}, the rate achieved $R^{(1)}$
during frame $2$ satisfies:
\begin{equation}
R^{(1)} - M\log\bigg( 1+ \frac{np
  \pi^2\sigma_s^2\sigma_d^2\eta\Ps}{32M^3(\eta\sigma_d^2+2\sigma_s^2\Ps+2)}
\bigg) \stackrel{w.p.1}{\longrightarrow} 0. \label{eq:R_{r1}}
\end{equation}

\subsubsection{Rate Achieved by $\mathcal{G}_2$}
During frame $2$, the relays in $\mathcal{G}_2$ receive the signal
vector:
\begin{align}
\rrx & = \sqrt{\frac{\Ps}{M}} \Hsr_2 \tx_2 + \Hrr \TR_1\bigg(
\sqrt{\frac{\Ps}{M}} \Hsr_1 \tx_1+
\mathbf{n}_1\bigg)+\mathbf{n}_2, \label{eq:rxr2}
\end{align}
where $\tx_2$ is the signal sent by the source during frame $2$, and
the second term corresponds to the interference from the transmission
of $\mathcal{G}_1$. During frame $3$ the relays in $\mathcal{G}_2$
transmit to the destination with a processing matrix
\begin{equation}
\TR_2 =\diag(T_{\frac{n}{2}+1}\,c_{\frac{n}{2}+1},\cdots,T_nc_n),
\end{equation}
where, to satisfy the power constraints, for $\frac{n}{2}+1\le i\le n$
\begin{equation}
c_i = e^{j\theta_i} \sqrt{\frac{\eta}{np(\Ps\sigma_s^2+\eta\sigma_r^2/2+1)}}.
\end{equation}
At the end of frame $3$, the received signal at the destination is
\begin{align}
\rx_2  = &\Hrd_2\TR_2\rrx + \mathbf{w}_2\nonumber
\\  = &\sqrt{\frac{\Ps}{M}} \underset{\Heq_2}{\underbrace{\Hrd_2 \TR_2
    \Hsr_2}} \tx_2  + \sqrt{\frac{\Ps}{M}} \Hrd_2\TR_2\Hrr\TR_1\Hsr_1\tx_1
\nonumber\\&
+ \underset{\weq_2}{\underbrace{\Hrd_2\TR_2\Hrr\TR_1 \mathbf{n}_1+ \Hrd_2 \TR_2 \mathbf{n}_2 +
    \mathbf{w}_2}} \quad . \label{eq:rxd2}
\end{align}

After correctly decoding $\tx_1$, the destination cancels the
inter-relay interference,\footnote{Interference cancellation requires
  knowledge of $\Hrr$ at the destination, however, we note that even
  without this knowledge it is possible to obtain the same scaling of
  secondary throughput with the number of relays. Intuitively, the
  inter-relay interference is bounded by a constant that is under our
  control.}
i.e., the second term in~\eqref{eq:rxd2}. After eliminating the
inter-relay interference, we have an equivalent channel:
\begin{equation}
\rx_2 = \sqrt{\frac{\Ps}{M}} \Heq_2 \tx_2 + \weq_2.
\end{equation}
Following steps similar to~\eqref{eq:HeqDiag}
and~\eqref{eq:am}, we have
\begin{equation}
\frac{\Heq_2}{\sqrt{n}}  \stackrel{w.p.1}{\longrightarrow} \rho_2 \bI, \label{eq:HeqARP}
\end{equation}
where $\rho_2=
\frac{\pi\sigma_s\sigma_d}{8M}\sqrt{\frac{p\eta}{\sigma_s^2\Ps+\eta\sigma_r^2/2+1}}$. Note
that $\weq_2$ is still a zero-mean Gaussian vector with
auto-covariance
\begin{equation}
\frac{1}{n} \auto_2 =
\Hrd_2\TR_2\Hrr\TR_1\TR_1^{\dag}\Hrr^{\dag}\TR_2^{\dag}\Hrd_2^{\dag} +
\Hrd_2\TR_2\TR_2^{\dag}\Hrd_2^{\dag} + \bI.\label{eq:WeqARP}
\end{equation}
In the right hand side of the above equation, we have
\begin{align}
\Hrr\TR_1\TR_1^{\dag}\Hrr^{\dag} & = \frac{\eta}{np(\Ps\sigma_s^2+1)}
\Hrr\,\diag(T_1,\cdots,T_{\frac{n}{2}})\,\Hrr^{\dag}
\nonumber\\ & \stackrel{w.p.1}{\longrightarrow}
\frac{\eta\sigma_r^2}{2(\Ps\sigma_s^2+1)} \bI.
\end{align}
Therefore,
\begin{equation}
\auto_2 \stackrel{w.p.1}{\longrightarrow} \lambda_2 \bI,
\end{equation}
where 
\begin{align}
\lambda_2=&\frac{1}{2(\Ps\sigma_s^2+\eta\sigma_r^2/2+1)}\bigg[\frac{\eta^2\sigma_d^2\sigma_r^2}{2(\Ps\sigma_s^2+1)}+\eta\sigma_d^2\bigg]
+1.
\end{align}
Combining~\eqref{eq:HeqARP} and~\eqref{eq:WeqARP} , the rate achieved
by $\mathcal{G}_2$ is $R^{(2)}$ where
\begin{equation}
R^{(2)} -  M \log\big( 1+ \frac{n\Ps \rho_2^2}{M \lambda_2}\big)  \stackrel{w.p.1}{\longrightarrow} 0. \label{eq:R_{r2}}
\end{equation}
The overall rate is given by the following theorem.
\begin{theorem}
\label{theorem:R_full}
Consider a secondary system with an $M$-antenna source, an $M$-antenna
destination, and $n$ single-antenna relays, in the presence of $N$
primary nodes each tolerating interference no more than $\gamma$. The
secondary rate satisfies
\begin{equation}
\overline{R} - \big( R^{(1)}+R^{(2)}\big)/3
\stackrel{w.p.1}{\longrightarrow} 0,
\end{equation}
under the Alternating Relaying Protocol with $\NumFrames=3$ and fixed
clustering.
\end{theorem}

From Theorem~\ref{theorem:R_full}, the growth rate of $\overline{R}$ is
\begin{equation}
\frac{\overline{R}}{\frac{2M}{3}\log n} \stackrel{w.p.1}{\longrightarrow} 1.
\end{equation}
\begin{remark}
\label{remark:R_full}
 Theorem~\ref{theorem:R_full} can be generalized to an arbitrary
 number of transmission blocks $\NumFrames$. For general $\NumFrames$ we can
 conclude:
\[
\frac{\overline{R}}{\frac{(\NumFrames-1)M}{\NumFrames}\log n} \stackrel{w.p.1}{\longrightarrow} 1.
\]
As $\NumFrames$ increases, the
growth rate of $\overline{R}$ approaches the maximum value of $M\log n$.
%% This increase in rate is significant for large $n$, but the same
%% improvements may not hold for small $n$.
\end{remark}

%%%%%%%%%%%%%%%%%%%%%%%%%%%%%%%%%%%%%%%%%%%%%%%%%%%%%%%%%%%%%%%%%%%
%%%%%%%%%%%%%%%%%%%%%%%%%%%%%%%%%%%%%%%%%%%%%%%%%%%%%%%%%%%%%%%%%%%
\section{Numerical Results}

Unless otherwise specified, we use parameters $\bar{\Ps}=10$, $M=2$,
$N=1$, $\gamma=5$ and $\sigma_s^2=\sigma_d^2=1$.

The secondary rates as a function of source transmit power are
presented by Figure~\ref{fig:Ps}. The theoretical rate under various
$\Ps$ is calculated according to~\eqref{eq:R^{(f)}}
and~\eqref{eq:R^{(g)}}. Recall that the theoretically optimal $\Ps$
given by~\eqref{eq:optPs} is obtained by~\eqref{eq:R^{(f)}}
and~\eqref{eq:R^{(g)}}. When the source interference links are very
weak, e.g., $\sigma_{sp}^2=0.1$, maximizing the source power is
optimal, which is similar to non-spectrum-sharing networks. When the
source interference links is strong, e.g., $\sigma_{sp}^2=1,2$, unlike
non-spectrum-sharing networks, the secondary achieves higher rate if
the source transmit at power lower than the maximum value. This is
because the source needs to ensure the relays can operate with
sufficient power, subject to the total interference constraints on the
primary nodes.

\begin{figure} 
\centering \includegraphics[width=3.5in]{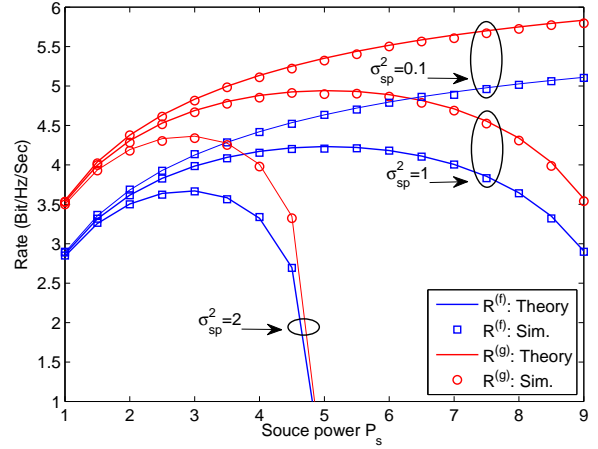}
\caption{Throughput as a function of source power when $n=100, \sigma_2=\sigma_d=\sigma_p=1$ }
\label{fig:Ps}
\end{figure}

Figure~\ref{fig:R} verifies Theorem~\ref{theorem:R_n} under fixed
clustering and gain clustering. Here, $\alpha=\gamma$,
$\sigma_{sp}^2=1$, $\Pr$ is given by~\eqref{eq:Pr} and $\Ps=5$, which
is almost optimal as shown in Figure~\ref{fig:Ps}. The simulated
average rate of $R_n$ under two clustering schemes are
compared to $R^{(f)}$ given by~\eqref{eq:R^{(f)}} and $R^{(g)}$ given
by~\eqref{eq:R^{(g)}}, respectively, where the results are well
matched for modest value of $n$. The secondary rate increases as the
interference links of relays become weaker (smaller $\sigma_p^2$),
since the relays can transmit at higher power (but the sum relay power
is still bounded with $n$).

\begin{figure}
\centering \includegraphics[width=3.5in]{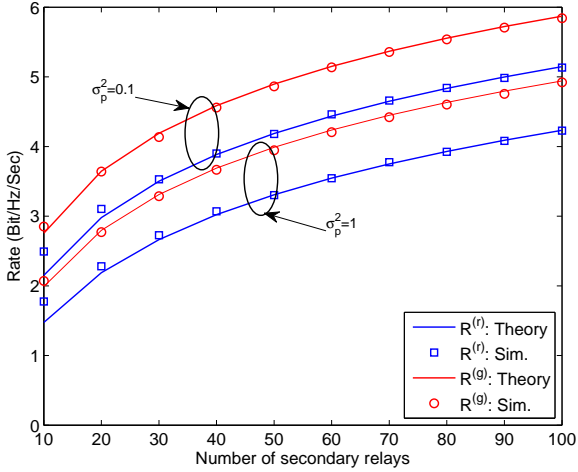}
\caption{Secondary rate under two clustering schemes }
\label{fig:R}
\end{figure}

Figure~\ref{fig:int} illustrates the tradeoff between maximizing
secondary rate and minimizing interference on the primary. The
interference power is $\gamma=5(n)^{-\delta}$ with $\delta=0.1$ and
$0.2$, respectively. For $\delta=0.2$, the interference power
decreases faster than $\delta=0.1$, while the secondary rate increases
more slowly. 
\begin{figure}
\centering
\includegraphics[width=3.5in]{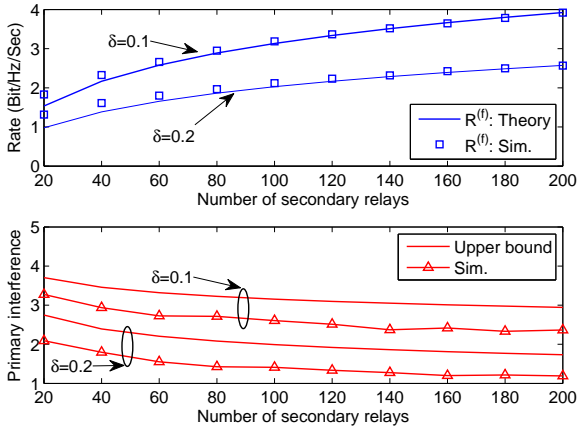}
\caption{Secondary rate and primary interference as a function of
  number of relays}
\label{fig:int}
\end{figure}

The rate of Alternating Relaying Protocol
(Theorem~\ref{theorem:R_full}) is shown in
Figure~\ref{fig:rateARP}. We consider $\alpha\rightarrow
\infty$, where all the relays in $\mathcal{G}_1$ and $\mathcal{G}_2$
transmit alternatively. Here, $\Ps=5$ and $\Pr$ is determined
by~\eqref{eq:IntConstraint3}. The simulated rates match the theoretic
analysis well under modest value of $n$. As the relay-relay channel
becomes weaker (smaller $\sigma_r^2$), the inter-relay interference is
reduced, and thus the secondary rate increases.

\begin{figure}
\centering
\includegraphics[width=3.5in]{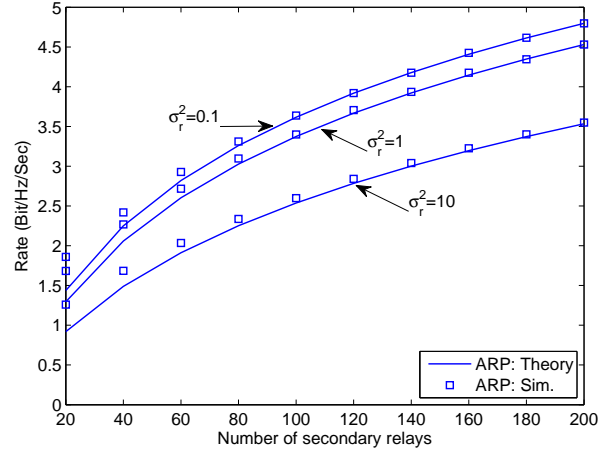}
\caption{Secondary rate under the alternating relaying protocol}
\label{fig:rateARP}
\end{figure}

%%%%%%%%%%%%%%%%%%%%%%%%%%%%%%%%%%%%%%%%%%%%%%%%%%%%%%%%%%%%%%%%%%%
%%%%%%%%%%%%%%%%%%%%%%%%%%%%%%%%%%%%%%%%%%%%%%%%%%%%%%%%%%%%%%%%%%%
\section{Conclusion}

This paper studies spectrum sharing networks with distributed AF
relaying to improve the secondary rate and reduce the interference on
the primary. In the asymptote of large $n$ (number of relays) the
optimal power strategy for the secondary source and relays was found,
achieving a secondary rate proportionally to $\log n$. The half-duplex
rate loss was reduced and the scaling of secondary rate was enhanced
by the introduction of the Alternating Relay Protocol. The trade-off
between the secondary rate and the interference on the primary was
characterized. Finally, our results show that even without cross
channel information at the secondary, the secondary rate can achieve
the growth rate $\log n$.

%%%%%%%%%%%%%%%%%%%%%%%%%%%%%%%%%%%%%%%%%%%%%%%%%%%%%%%%%%%%%%%%%%%
\bibliographystyle{IEEEtran} \bibliography{IEEEabrv,LiYang-final}

\end{document}